\documentclass[a4paper,11pt]{amsart}
\usepackage[utf8]{inputenc}
\usepackage{amssymb,amsthm,enumerate,bm,bbm}
\usepackage[pdfborder={0 0 0}]{hyperref}
\usepackage[foot]{amsaddr}

\newcommand{\Ex}{\mathbb{E}}
\newcommand{\Var}{\mathrm{Var}}

\newcommand{\Cov}{\mathrm{Cov}}
\newcommand{\dur}{\mathrm{Dur}}
\newcommand{\excdur}{\mathrm{ExcDur}}

\newcommand{\bz}{\bar z}
\newcommand{\bfi}{\bar f_{\infty}}
\newcommand{\bD}{\bar D}
\newcommand{\bP}{\bar P}
\newcommand{\Dz}{\Delta z}
\newcommand{\one}{\mathbbm{1}}
\newcommand{\mT}{\mathcal{T}}
\newcommand{\mE}{\mathcal{E}}

\renewcommand{\phi}{\varphi}

\title{How to hedge extrapolated yield curves}
\author{Andreas Lagerås}
\address{AFA Insurance\\106 27 Stockholm\\Sweden}
\thanks{The opinions expressed in this paper are not necessarily those of AFA Insurance.}
\email{andreas@math.su.se}
\keywords{Term structure, Calculus of variations, Functional analysis}
\subjclass[2010]{91G80, 
                 46N10, 
                 91G30
}

\newtheorem*{thm}{Theorem}
\newtheorem{prop}{Proposition}

\theoremstyle{remark}
\newtheorem{ex}{Exemple}

\begin{document}

\maketitle

\begin{abstract}
We present a framework on how to hedge the interest rate sensitivity of liabilities discounted by an extrapolated yield curve. The framework is based on functional analysis in that we consider the extrapolated yield curve as a functional of an observed yield curve and use its G\^ateaux variation to understand the sensitivity to any possible yield curve shift. We apply the framework to analyse the Smith-Wilson method of extrapolation that is proposed by the European Insurance and Occupational Pensions Authority (EIOPA) in the coming EU legis\-lation Solvency II, and the method recently introduced, and currently prescribed, by the Swedish Financial Supervisory Authority.
\end{abstract}

\section{Introduction}

Insurance companies, especially life assurance companies, can have liabilities further into the future than there exists a liquid market for fixed income financial assets. These liabilities can therefore not be given a pure market value, but must be discounted by a yield that is to some extent model based and \emph{extrapolated} from market yields beyond some \emph{last liquid point} (LLP).

This issue is related to the fact that one typically wants a yield curve for a continuum of times to maturity, whereas only a discrete number of financial instruments are used for deriving the curve. Furthermore, the price of a zero coupon bond is the only directly observable true market discount factor, and zero coupon bonds are not that common. One therefore has to \emph{bootstrap} a yield curve even for time to maturities shorter than the LLP.

The extrapolation method is sometimes essentially the same as the bootstrap method, and the framework presented in this paper can be used to analyse them both from the same point of view, viz.\ that we want to know the sensitivity of a discount factor for a given time to maturity with respect to \emph{all} market rates that are used to build the discount curve. We will however focus on extrapolation and not bootstrapping in itself.

The idea is to compute the total differential of the discount factor with respect to the prices of market instruments used to build the curve. This is essentially the same thing as computing the key rate durations of a liability discount factor. This could become unwieldy as the differential would have as many terms as the number of market instruments used. To obtain \emph{qualitative} results on different extrapolation methods, we consider an idealised case where a continuum of zero coupon bonds are used for curve construction. The differential in this case is replaced with the G\^ateaux variation, and the sum in the discrete case is in general replaced by an integral which turns out to be easy to interpret.

We apply this framework to some simple extrapolation schemes, among them the method prescribed by Swedish Financial Supervisory Authority (SFSA) \cite{fi}, and the so called Smith-Wilson method prescribed by European Insurance and Occupational Pensions Authority (EIOPA) for the coming EU wide Solvency II regulations, \cite{eiopa,sw}. 

Since we are mostly interested in qualitative results we will use continuously compounded rates when we describe the methods, even where the legislation might use annual compounding. We limit the our investigation to the case of instantaneous changes in the market yield curve. Changes over longer time spans are also of importance in practice since hedges need to be reset.

Qualitatively then, the SFSA method makes all liabilities beyond the last market observation sensitive to the zero coupon yield at that observation, whereas the Smith-Wilson method makes them sensitive to both the zero coupon yield and the forward rate at the last market observation. The dependency on the forward rate has been noted by other researchers in the special case of annually spaced market rates for ``typical'' shapes of the market yield curve, see e.g.\ \cite{cardano}. We show how this is an intrisic feature of the method regardless of the shape of the market yield curve. The dependency on the last forward rate is unfortunate as it would be very hard to initiate a hedge in any substantial size, since exposure to a forward rate is replicated by shorting one bond and going long another.

The main contributions of this paper are twofold. Firstly, the framework provides a straightforward way to compute the optimal hedge of liabilities with regards to an extrapolation method. Secondly, the composition of the optimal hedge allows for a discussion of whether a given extrapolation method is feasible for the individual insurance company and for the financial market as a whole.

This paper proceeds with an introduction to the general theory and framework in Section \ref{sec:theory}, and we then present some extrapolation methods in Section \ref{sec:methods}. In Section \ref{howto}, we apply the theory to the methods and derive the optimal hedges, if possible, and in Section \ref{sec:conclusions} we discuss the results and outlooks to future research.

\section{Theory}\label{sec:theory}

\subsection{Discount factors and yields}

Let $y$ be a generic yield curve and let $D_t:=e^{-ty_t}$ be the discount factor for time to maturity $t$ with $y$ as discount curve. We write $D_t[y]$ to stress that it is a function of $y$. We use brackets to highlight arguments that are themselves functions, i.e.\ $D_t[y]$ can be called a \emph{functional}. 

Our analysis will mostly be static, in that we will consider instantaneous changes in a yield curve and its consequences. We take the present time to be 0 so that we can refer to ``time to maturity'' simply as ``maturity'' or ``time''. The unit of time is usually taken to be years, and we will restrict our analysis to times $t\in\mT:=[0,T]$ where $T$ is arbitrary but fixed, say $T=200$ to cover most imaginable insurance cash flows.

We let $z$ denote the market zero coupon bond curve so that $D_t[z]$ is the price of a zero coupon bond with maturity $t$, and let $\bz$ denote the discount curve to be used for valuing liabilities, i.e.\ the present value of one unit of liability at time $t$ is $D_t[\bz]$. We will consider several cases where $\bz$ is a function of $z$ and we write $\bz[z]$ to stress this. We also define $\bD_t[z]:=D_t[\bz[z]]$ --- or in other notation $\bD := D\circ \bz$ --- for the liability discount factor as a function of market rates.

The general theory only uses the (zero coupon) yields, but for some applications we need forward rates. We define the market (instantaneous) forward rate $f_t:=\frac{d}{dt}(tz_t)$ and the discount forward rate $\bar f_t:=\frac{d}{dt}(t\bar z_t)$, so that
\begin{equation}\label{z_from_f}
z_t=\frac{1}{t}\int_0^tf_s\,ds,
\end{equation}
and similarly for $\bz$ and $\bar f$.

Since we intend to do some functional analysis with yield and discount curves we must decide on a space of functions for them. In general we assume that $z\in\mathcal{C}_{\mathrm{s}}(\mT)$, the space of cadlag functions with at most a finite number of jumps, with norm $\|z\|:=\sup_{t\in\mT}|z_t|$. When we need the existence of forward rates we assume that $z\in\mathcal{C}_{\mathrm{s}}^1(\mT)$, the space of of cadlag functions with at most a finite number of jumps and with first derivatives, with norm $\|z\|:=\sup_{t\in\mT}|z_t|+\sup_{t\in\mT}|z'_t|$. Both these are normed linear spaces, see \cite[Ch.\ 1.3]{sagan}. One could probably choose larger spaces, but these suffice for the applications we have in mind.

Furthermore, we assume that $\bz$ is defined on the whole of $\mathcal{C}_{\mathrm{s}}(\mT)$ --- or $\mathcal{C}_{\mathrm{s}}^1(\mT)$ when forward rates are needed --- and has range in the respective space.

\subsection{Cash flows and present values}

We represent a generic cash flow with a function $C$, where $C_t$ is defined as the cumulative cash flow in the interval $[0,t]$. Note that the function $C$ has a jump at $t$ if there is a lump sum payment at time $t$. We will somewhat sloppily refer to the cash flow represented by $C$ simply as ``the cash flow $C$'' or ``$C$''. We insist on $C$ being of bounded variation, i.e.\ that it can be written as a difference of two non-decreasing functions, $C=C^+-C^-$. This seems reasonable with the two terms representing, say, inflows ($C^+$) and outflows ($C^-$). The present value of $C$ discounted by the yield curve $y$ is given by the Stieltjes integral
$$
P[y;C] := C_0+\int_\mT D_t[y]\,dC_t.
$$
The integral is well defined since $C$ is of bounded variation.

We define $C^*_t[y] := C_0+\int_0^t D_s[y]dC_t$ as the present value of the cash flow up to time $t$ so that $P[y;C]=C_0+\int_\mT dC^*_t[y] = C^*_T[y]$ and, at least informally, $dC^*_t[y]=D_t[y]dC_t$ is the present value of the cash flow at time $t$. We will drop the argument in the notation when it can be inferred from the context. Typically, $A^*_t = A^*_t[z]$ and $L^*_t=L^*_t[\bz]$.

The market value of asset cash flows represented by the function $A$ and the discounted values of liabilities represented by the function $L$ are thus $P[z;A]$ and $P[\bz;L]$ respectively. We define $\bar P[z;L] := P[\bz[z];L]$, i.e.\ the discounted value of the liabilities as a function of $z$.

\subsection{Hedging}

Let $L$ be a liability cash flow. We say that the asset cash flow $A$ is a \emph{perfect hedge} of $L$ if for all $z$
\begin{equation}\label{P_bP}
P[z;A]=\bar P[z;L].
\end{equation}
We say that $\bz$ is \emph{perfectly hedgeable} if there for all $L$ exists a perfect hedge.

$A$ is a \emph{first order hedge} of $L$ (at $z$) if
$$
P[z+\epsilon\Dz;A]-P[z;A]=\bP[z+\epsilon\Dz;L]-\bP[z;L]+o(\epsilon).
$$

Note that the definition of first order hedge is contingent on $z$: $A$ might be a first order hedge of $L$ at one $z$ but not another. Also note that a first order hedge at $z$ does not necessarily have $P[z;A]=\bar P[z;L]$: it is any possible change in present value of the liabilities --- as the market yield curve changes from $z$ to $z+\epsilon\Dz$ --- that is matched by the assets (up to a remainder small in $\epsilon$), not the present value itself.

We say that $\bz$ is \emph{first order hedgeable} if there for all $L$ and at all $z$ exists a first order hedge.

\subsection{Functional derivatives and Taylor approximation}\label{sec:gateaux_taylor}

In this subsection \ref{sec:gateaux_taylor}, $f$ will denote a generic functional and not necessarily the forward curve.

The \emph{G\^ateaux variation}, or simply the \emph{variation}, of a function $f[g]$ in the direction $h$ is defined by
\begin{equation}\label{Gv_def}
\delta f[g|h] := \lim_{\epsilon\to 0^+}\frac{f[g+\epsilon h]-f[g]}{\epsilon}.
\end{equation}

The variation is homogenous of degree one in $h$ and we will use this for Taylor approximation, i.e.\ if $g$ changes to $g+\Delta g$ we have, see \cite[Thm 1.5]{sagan},
\begin{equation}\label{taylor1}
f[g+\Delta g] = f[g] + \delta f[g|\Delta g] + R(\Delta g),
\end{equation}
where $\lim_{\epsilon\to 0^+}R(\epsilon \Delta g)/\epsilon = 0$.

Note that $h$ in \eqref{Gv_def} and $\Delta g$ in \eqref{taylor1} in general can be \emph{functions} (or functionals) themselves. We clearly need $g+\Delta g$ to be in the domain of $f$ in \eqref{taylor1}, and we will in the following always assume that the shift, $\Delta g$ in this case, is such that this criterion is fulfilled, e.g.\ by having a small enough norm.

The following chain rule holds
\begin{equation}\label{chain}
\delta (f \circ g)[h|k] = \delta f\big[g[h]\big|\delta g[h|k]\big].
\end{equation}

\textbf{Remark.} There are more restrictive versions of derivatives on function spaces that one could use. For example, if one restrics $\delta f[g|h]$ to be linear and bounded in $h$ one get the so-called \emph{G\^ateaux differential}, and if furthermore the remainder $R$ in the Taylor expansion \eqref{taylor1} tends to 0 uniformly in $\Delta g$ and not only along the ray $\{\epsilon\Delta g:\epsilon>0\}$, i.e.\ $\lim_{\|\Delta g\|\to 0}R(\Delta g)/\|\Delta g\| = 0$, one get the so-called \emph{Fr\'echet differential}. The linear functional $\delta f[g|\,\cdot\,]$ is called the \emph{G\^ateaux} and \emph{Fr\'echet derivative}, in the respective cases --- see \cite{nashed}. 

Example \ref{Gv_nec} below shows why the linearity of the differential cannot be taken for granted in the applications we consider, and that we therefore need the generality afforded by the G\^ateaux variation.

\begin{ex}
The variation of a discount factor $D_t[y]$ in the direction $\Delta y$ is $\delta D_t[y|\Delta y]=-t\Delta y_tD_t[y]$ since
\begin{align*}
\delta D_t[y|\Delta y] &= \lim_{\epsilon\to 0^+}\frac{D_t[y+\epsilon\Delta y]-D_t[y]}{\epsilon}=\lim_{\epsilon\to 0^+}\frac{e^{-t(y_t-\epsilon\Delta y_t)}-e^{-ty_t}}{\epsilon} \\
&=-t\Delta y_t e^{-ty_t} = -t\Delta y_t D_t[y].
\end{align*}
\end{ex}

\begin{ex}
For the present value of a cash flow we have for a general direction $\Delta y$
\begin{align}
\delta C^*_T[y|\Delta y] &= \delta P[y;C|\Delta y] = \int_\mT \delta D_t[y|\Delta y]\,dC_t \notag\\
&= -\int_\mT t\Delta y_t D_t[y]\,dC_t = -\int_\mT t\Delta y_t\,dC_t^*[y] \label{diffP}
\end{align}
In the special case where $\Delta y$ is a constant function, say $\Delta y_t=c$ for all $t$, we get
$$
\delta C^*_T[y|\Delta y] = -c\int_\mT t\,dC^*_t[y]
$$
With $c=-1$ we call this quantity the \emph{dollar duration} of $C$ (at $y$). We call it dollar duration regardless of in what currency $C$ is denominated. In practice one often consider the case $c=-0.0001$ (a \emph{basis point}), and talk of the \emph{dollar value of a basis point} (DV01). The \emph{duration} of $C$ is defined as the dollar duration divided by the present value:
$$
\dur[y,C] := \frac{\int_\mT t\,dC^*_t[y]}{C^*_T[y]}. 
$$
\end{ex}

\begin{ex}\label{Gv_nec}
Let $\bz_t[z]:=\max(0,z_t-c)$ for some positive constant $c$, i.e.\ the liability discount rate is equal to the market rate adjusted downwards with $c$ and put to 0 if this difference is negative. Similar constructions abound in various ``stress tests'' of liability discount rates that supervisory agencies use. By the definition \eqref{Gv_def} we get
\begin{align*}
\delta\bz_t[z|\Delta z] &= \lim_{\epsilon\to 0^+}\frac{\max(0,z_t+\epsilon\Delta z_t-c)-\max(0,z_t-c)}{\epsilon} \\
&= \begin{cases}
0, & z_t < c \\
0, & z_t = c\text{ and }\Delta z_t \leq 0 \\
\Delta z_t, & z_t = c\text{ and }\Delta z_t > 0 \\
\Delta z_t, & z_t > c.
\end{cases}
\end{align*}
The G\^ateaux variation $\delta\bz[z|\Delta z]$ is clearly not linear in $\Delta z$ at $z$ if $z_t=c$ for some $t$.
\end{ex}

The second order G\^ateaux variation of $f$ in directions $h$ and $k$ is
$$
\delta^2 f[g|h,k] := \lim_{\epsilon\to 0^+}\frac{\delta f[g+\epsilon k|h]-\delta f[g|h]}{\epsilon},
$$
and to ease notation we write $\delta^2 f[g|h]:=\delta^2 f[g|h,h]$ when both directions are the same. This will be used for second order Taylor approximation:
$$
f[g+\Delta g] = f[g] + \delta f[g|\Delta g] + \frac12 \delta^2 f[g|\Delta g] + R_2(\Delta g),
$$
with $\lim_{\epsilon\to 0^+}R_2(\epsilon \Delta g)/\epsilon^2 = 0$.

The chain rule for the second order variation reads
\begin{align*}
\delta^2(f\circ g)[h|k,l] &= \delta^2 f\big[g[h]\big|\delta g[h|k],\delta g[h|l]\big] +  \delta f\big[g[h]\big|\delta^2 g[h|k,l]\big],
\end{align*}
or with both directions the same,
$$
\delta^2(f\circ g)[h|k] = \delta^2f\big[g[h]\big|\delta g[h|k]\big] +  \delta f\big[g[h]\big|\delta^2 g[h|k]\big].
$$

\begin{ex}
The second order variation of a discount factor $D_t[y]$ in the direction $\Delta y$ is $\delta^2D_t[y|\Delta y]=t^2(\Delta y_t)^2D_t[y]$. For a present value of a cash flow we have $\delta^2 C^*_T[y|\Delta y]=\int_\mT t^2(\Delta y_t)^2\,dC^*_t$. When the yield curve $y$ shifts in a parallel fashion, i.e.\ $\Delta y_t=c$ for some constant $c$ for all $t$, $\delta^2 C^*_T[y|\Delta y]=c^2\int_\mT t^2\,dC^*_t$. The quantity $\int_\mT t^2\,dC^*_t\big/C^*_T$ is called the \emph{convexity} of $C$ (at $y$).
\end{ex}

\begin{ex}
By the chain rule, the second order variation of $\bD_t[z]$ in the direction $\Dz$ is
$$
\delta^2\bD_t[z|\Dz] = \left(t^2(\delta\bz_t[z|\Dz])^2 - t\delta^2\bz_t[z|\Dz]\right)\bD_t[z]
$$
and thus
\begin{equation}\label{diff2P}
\delta^2\bP[z;L|\Dz] = \int_\mT \left(t^2(\delta\bz_t[z|\Dz])^2 - t\delta^2\bz_t[z|\Dz]\right)dL^*_t.
\end{equation}
\end{ex}

\subsection{General results}

Let us first consider perfect hedges. The following proposition shows that perfect hedgeability is equivalent to the liability discount factor being affine in the market discount factors. 

\begin{prop}\label{prop:perfect_hedgeability}
$\bz$ is perfectly hedgeable if and only if for all $t$, $\bD_t[z]=P[z;C^{(t)}]$ for some cash flow $C^{(t)}$ independent of $z$. The perfect hedge has the form $A_s=L_0+\int_\mT C^{(t)}_s\,dL_t$.
\end{prop}
\begin{proof} To prove necessity, we want to show that perfect hedgeability implies the existence of the cash flow $C^{(t)}$ of the theorem. Consider the liability cash flow $L_s^{(t)}:=\one\{s\geq t\}$, representing a lump sum of 1 at time $t$, and let $A^{(t)}$ be the perfect hedge of $L^{(t)}$. By the definition of a discount factor and the definition of a perfect hedge, \eqref{P_bP}, we have $\bD_t[z] = \bP[z;L^{(t)}] = P[z;A^{(t)}]$, and we can let $C^{(t)}=A^{(t)}$. Also note that $A^{(t)}_s=L_0^{(t)}+\int_\mT C^{(u)}_sdL^{(t)}_u$.

To prove sufficiency, we want to show that the asset cash flow $A$ of the theorem is a perfect hedge. 
\begin{align*}
P[z;A] &= A_0 + \int_{s\in\mT}D_s[z]dA_s \\
&= L_0+\int_{t\in\mT}C^{(t)}_0\,dL_t + \int_{s\in\mT} D_s[z]\int_{t\in\mT} dC_s^{(t)}\,dL_t \\
&= L_0+\int_{t\in\mT}\bigg(C^{(t)}_0+\int_{s\in\mT} D_s[z]\,dC_s^{(t)}\bigg)\,dL_t \\
& = L_0+\int_\mT P[z;C^{(t)}]\,dL_t \\
&= L_0+\int_\mT \bD_t[z]\,dL_t = \bP[z;L].
\end{align*}


\end{proof}

We now turn to first order hedgeability.
\begin{prop}\label{prop:first_order_hedgeability}
$\bz$ is first order hedgeable at $z$ if and only if $\delta P[z;A|\Dz] = \delta \bP[z;L|\Dz]$, provided the G\^ateaux variation exists, and in that case the first order hedge $A$ solves
\begin{equation}\label{hedgeeq}
\int_\mT t\Dz_t\,dA^*_t = \int_\mT t\delta\bz_t[z|\Dz]\,dL^*_t
\end{equation}
\end{prop}
\begin{proof}
Applying the Taylor approximation \eqref{taylor1} to $P[z+\Dz;A]$ and $\bP[z+\Dz;L]$ gives us
\begin{align*}
P[z+\Dz;A] - P[z;A] &= \delta P[z;A|\Dz] + R(\Dz) \\
\bP[z+\Dz;L] - \bP[z;L] &= \delta \bP[z;L|\Dz] + R(\Dz), 
\end{align*}
and first order hedgeability is thus equivalent to $\delta P[z;A|\Dz] = \delta \bP[z;L|\Dz]$. Equation \eqref{hedgeeq} follows from equation \eqref{diffP} and the chain rule \eqref{chain}.
\begin{align*}
-\delta P[z;A|\Dz] &= \int_\mT t\Dz_t\,dA_t^* \\
-\delta \bP[z;L|\Dz] &= -\int_\mT \delta \bD_t[z|\Dz]\,dL_t = -\int_\mT \delta D_t\big[\bz[z]\big|\delta\bz[z|\Dz]\big]\,dL_t \\
&= \int_\mT t\delta\bz_t[z|\Dz] D_t[\bz[z]]\,dL_t = \int_\mT t\delta\bz_t[z|\Dz]\,dL_t^* 
\end{align*}
\end{proof}
We will use equation \eqref{hedgeeq} repetedly and call it the \emph{hedge equation}.

If $A$ is a first order hedge of $L$, the second order Taylor expansion can be used to understand how well the hedge performs. By equation \eqref{diff2P},
\begin{align}
P[z&+\Dz;A]-P[z;A]-(\bP[z+\Dz;L]-\bP[z;L])\notag\\
&= \frac{1}{2}\left(\delta^2 P[z;A|\Dz] - \delta^2\bP[z;L|\Dz]\right) + R_2(\Dz) \notag\\
&= \frac{1}{2}\bigg(\int_\mT t^2(\Dz_t)^2\,dA^*_t \notag\\
&\qquad\quad - \int_\mT \left(t^2(\delta\bz_t[z|\Dz])^2 - t\delta^2\bz_t[z|\Dz]\right)\,dL^*_t\bigg) + R_2(\Dz). \label{hedge_cnvx}
\end{align}

The sensitivity of the present value of liabilities with respect to the parameters of $\bz$, is also of interest. If $\theta$ is a scalar parameter,
\begin{equation}\label{param_sens}
\frac{d}{d\theta}P[\bz;L]=\int_\mT \frac{d}{d\theta}D_t[\bz]\,dL_t = -\int_\mT t\frac{d\bz_t}{d\theta}\,dL^*_t.
\end{equation}

\section{Some extrapolation methods}\label{sec:methods}

Here we will describe six possible extrapolation methods. As noted in the introduction we will use continuously compounded interest rates even if the methods prescibed by law might use annual compounding. We describe these methods in the idealised case where there zero coupon bond prices are available for all maturities up to some specified time.

Methods 1 and 2 are described starting from zero coupon yields, whereas the methods 3 and 4 are described starting from forward rates. Methods 1 and 3 do not really extrapolate the respective type of curve but set the long term yield or forward rate to a predetermined constant value. Methods 2 and 4 are constant extrapolation of the respective type of curve. Method 5 is prescibed by the Swedish Financial Supervisory Authority (SFSA), and method 6, called the Smith-Wilson method, is suggested by European Insurance and Occupational Pensions Authority (EIOPA) to be used under Solvency II.

\addtocounter{subsection}{-1}
\subsection{Commonalities}

All the methods have $\bar z_t=z_t$ for $t\leq \tau$. The time to maturity $\tau$ is sometimes called the \emph{last liquid point} (LLP). The methods differ in their expressions for $\bz_t$ with $t\in\mE:=(\tau,T]$, i.e.\ in the extrapolated part of the yield curve.

Methods 1, 3, 5, and 6 have a predetermined limiting value for the forward rate, $\bar f_{\infty}:=\lim_{t\to\infty}\bar f_t$, called the \emph{ultimate forward rate} (UFR). Note that the existence of the limit $\bar f_{\infty}$ implies that $\lim_{t\to\infty}\bar z_t=\bar f_{\infty}$, though the convergence to the limiting value is slower for $\bar z$ than for $\bar f$.

Methods 5 and 6 both have an additional parameter $\kappa>\tau$ that is interpreted as the time to maturity where the extrapolated forward rate should equal or be close enough to the UFR. Details follow in the respective sections below.

Many proposed regulatory methods actually add a constant $c$, or a constant curve $c_t$, to the market yield curve $z$ before it is used to derive the bootstrapped and extrapolated curve, i.e.\ one uses $\bz[z+c]$ instead of $\bz[z]$.

The constant $c$ is introduced to adjust the market rates for issues such as (a.) the zero rates $z$ might have a credit risk component if they are based on swaps ($c<0$), and (b.) the supervisory agency might want to give insurance companies some relief ($c>0$).

This implies that $\bar D_t = e^{-tc_t}D_t$ for $t\leq \tau$, so the cash flows --- or parts of cash flows --- that were perfectly hedgeable when $c$ had not been introduced are still perfectly hedgeable. For first order hedgeability, we have to consider shifts $\Dz$ of the yield curve $z+c$ rather than $z$, and since $\delta\bz[z+c|\Dz]=\delta\bz[z|\Dz]$ the important part of the hedge equation \eqref{hedgeeq} is unaffected.

The sensitivity of liabilities with respect to changing $c$, say from $c$ to $c+\Delta c$, is also quite transparent since we can apply \eqref{diffP} with $\Delta c$ substituted for $\Delta y$.

We therefore proceed with $c=0$. 

\subsection{Method 1. Predetermined long term zero coupon yields.}

This is not really an extrapolation method as all zero coupon yields beyond $\tau$ are set to a constant. It is nevertheless useful as a baseline method.
$$
\bar z_t = \bar f_{\infty},\quad t\in\mE,
$$
which implies $\bar f_t=\bar f_{\infty}$ and $\bar D_t = e^{-\bar f_\infty t}$. Note that $\bar z$ has a discontinuity at $\tau$ unless $z_{\tau}$ happens to equal $\bar f_{\infty}$.

\subsection{Method 2. Constant extrapolation of zero coupon yields.}

Here
$$
\bar z_t = z_\tau,\quad t\in\mE,
$$
which implies $\bar f_t = z_\tau$ and $\bar D_t = e^{-t z_{\tau}}=D_{\tau}^{t/\tau}$. 
This method gives a discount curve that is continuous at $\tau$, though it might have a kink, i.e.\ a discontinuous first derivative, at $\tau$.

\subsection{Method 3. Predetermined long term forward rates}

This method is similar to Method 1 though it introduces constant forward rates beyond $\tau$ rather than constant zero coupon rates.
$$
\bar f_t = \bar f_{\infty},\quad t\in\mE,
$$
which implies
$$
\bar z_t = \frac{\tau}{t}z_\tau + \left(1-\frac{\tau}{t}\right)\bar f_\infty,
$$
and $\bar D_t = e^{-\tau z_{\tau}-(t-\tau)\bar f_\infty}=e^{-(t-\tau)\bar f_\infty}D_\tau$. 

Note that even if the forward curve $\bar f$ is discontinuous at $\tau$, the discount curve $\bar z$ is not. If $\bar f$ is discontinuous at $\tau$, the discount curve $z$ will have a kink at $\tau$.

\subsection{Method 4. Constant extrapolation of forward rates}

This method is similar to Method 2 though it constantly extra\-polates forward rates beyond $\tau$ rather than zero coupon rates.
$$
\bar f_t = f_\tau,\quad t\in\mE,
$$
which implies
$$
\bar z_t = \frac{\tau}{t}z_\tau + \left(1-\frac{\tau}{t}\right)\bar f_\tau.
$$
and $\bar D_t = e^{-\tau z_{\tau}-(t-\tau)\bar f_\tau}$.  

Since the forward curve $\bar f$ is continuous at $\tau$, $\bar z$ has no kink there.

\subsection{Method 5. SFSA}\label{method_5}

This method is prescribed for Swedish insurance companies by the SFSA \cite{fi} and it is an elaboration on Method 3, where predetermined long term discount forward rate $\bar f_{\infty}$ is phased in linearly between $\tau$ and $\kappa$.
$$
\bar f_t := \begin{cases} \frac{\kappa-t}{\kappa-\tau}f_t + \frac{t-\tau}{\kappa-\tau}\bar f_{\infty}, & \tau < t \leq \kappa, \\
                          \bar f_{\infty}                                                  & t > \kappa.        \end{cases}
$$
We show in Appendix \ref{z_method_5} that
\begin{align*}
\bz_t &= \begin{cases} \frac{\kappa-t}{\kappa-\tau}z_t + \frac{1}{t}\frac{1}{\kappa-\tau}\int_{\tau}^t sz_sds + \frac{t-\tau}{\kappa-\tau}\big(1-\frac{\tau}{t}\big)\frac{\bar f_{\infty}}{2}, & \tau < t \leq \kappa, \\
                     \frac{1}{t}\frac{1}{\kappa-\tau}\int_{\tau}^{\kappa} sz_sds + \big(1-\frac{\tau+\kappa}{2t}\big)\bar f_{\infty}, & t > \kappa.        \end{cases} 
\end{align*}
Since the forward curve $\bar f$ is continuous at $\tau$ and $\kappa$, $\bar z$ has no kinks there.

\subsection{Method 6. Smith-Wilson}

This method has been suggested by EIOPA \cite{eiopa} and is usually described in terms of inter\-pola\-tion and extra\-pola\-tion of a finite number of discount factors \cite{eiopa,sw}. For a curve built from the market rates at time to maturities $t_1,\dots,t_N$,
$$
\bar D_t := e^{-\bfi t}+\sum_{i=1}^NW(t,t_i)\zeta_i,
$$
where
$$
W(s,t) := e^{-\bfi (s+t)}\big(\alpha\min(s,t)-e^{-\alpha\max(s,t)}\sinh(\alpha\min(s,t))\big),
$$
and where $\bm{\zeta}:=(\zeta_1,\dots,\zeta_N)$ is determined by $\bar D_{t_i} = D_{t_i}$ for $i=1,\dots,N$. As written here, this model has a free parameter $\alpha> 0$. This parameter governs the speed of convergence for the forward rates towards the UFR; the higher the $\alpha$, the faster the convergence. The actual EIOPA method requires that $\alpha$ shall be set to ensure that $|\bar f_{\kappa}-\bar f_{\infty}|$ is less than or equal to a specified value $\epsilon$ with $\kappa>\tau$. If $f_{\tau}$ is close enough to $\bfi$, $\alpha$ is not well-defined.

A comprehensive analysis of the case where $\alpha$ is defined by a convergence criterion is beyond the scope of this paper, but we will indicate the necessary steps in that direction.

We show in Appendix \ref{z_method_6} that the continuous version of this method, when market observations are used up to $\tau$, has the discount factor
$$
\bar D_t = e^{-\bar f_{\infty}(t-\tau)}D_{\tau}\bigg(1+(\bfi -f_{\tau})\frac{1-e^{-\alpha(t-\tau)}}{\alpha}\bigg),\quad t\in\mE.
$$
Note that $\frac{1-e^{-\alpha(t-\tau)}}{\alpha}$ is increasing from 0 to $\frac{1}{\alpha}$ as $t$ goes from $\tau$ to $\infty$. Thus, the discount factor will become \emph{negative} for high enough values of $t$ unless $f_{\tau} \leq \bfi +\alpha$. This problem with the Smith-Wilson method has also been noted by others, e.g.\ Rebel \cite{cardano}.

Provided then that $f_{\tau} \leq  \bfi +\alpha$, we have
\begin{equation}\label{sw_z}
\bar z_t = \frac{\tau}{t}z_{\tau}+\left(1-\frac{\tau}{t}\right)\bar f_{\infty}-\frac{1}{t}\log\left(1+(\bar f_{\infty}-f_{\tau})\frac{1-e^{-\alpha(t-\tau)}}{\alpha}\right),
\end{equation}
and
\begin{equation}\label{sw_f}
\bar f_t = \bar f_{\infty} - \frac{(\bar f_{\infty}-f_{\tau})e^{-\alpha(t-\tau)}}{1+(\bar f_{\infty}-f_{\tau})\frac{1-e^{-\alpha(t-\tau)}}{\alpha}}.
\end{equation}

A zero coupon yield curve is arbitrage free if and only if the corresponding forward curve is non-negative. The previously described extrapolation methods are clearly arbitrage free if the market curve $z$ is arbitrage free and $\bar f_{\infty}\geq 0$.

For the Smith-Wilson curve we can do the following analysis. Since $e^{-\alpha(t-\tau)}$ is decreasing from 1  as $t$ increases from $\tau$, the forward rate $\bar f_t$ tends mono\-ton\-ous\-ly toward $\bar f_{\infty}$ as $t$ increases from $\tau$. This means that forward rates $\bar f_t$ are non-negative for all $t\in\mE$ provided that $f_{\tau}\geq 0$ and $\bar f_{\infty}\geq 0$. The Smith-Wilson method thus provides an arbitrage free extrapolated yield curve in this continuous setting. The Smith-Wilson curve is however \emph{not} necessarily arbitrage free in the case when it is fitted to a finite number of market yields, see Appendix \ref{arb}.

%
%
%
%
%
%
%
%

\section{How to hedge}\label{howto}

We proceed by analysing each method of Section \ref{sec:methods}. In order to apply the hedging equation \eqref{hedgeeq} we need the Gâteaux variation $\delta\bar z[z|\Delta z]$. If the zero coupon yield curve changes from $z$ to $z+\Delta z$, then the forward rate curve changes from $f$ to $f+\Delta f$, where $\Delta f_t := \frac{d}{dt}(t\Delta z_t)$, since $f$ is linear in $z$.

Since all methods have a perfect hedge for cash flows at times $t\leq\tau$, we assume that $L_\tau=0$ in order to focus on cash flows for times $t\in \mE$. We also assume $L_T>0$ to avoid trivialities.

\subsection{(Not) hedging Method 1}

We recall that $\bar z_t:= \bz_{\infty}$ for $t>\tau$, and thus $\delta\bz_t[z|\Dz] = 0$. The hedge equation reads (recall that $L_{\tau}=0$),
\begin{align*}
\int_\mE t\Dz_t\,dA^*_t = \int_\mE t\delta\bar z_t[z|\Dz]\,dL^*_t = 0
\end{align*}
and this holds if $A_t=0$ for all $t$, i.e.\ the ``extrapolated'' part of the discount curve is \emph{not} hedged. This is reasonable since it is not sensitive to changing market rates. This is also a perfect hedge according to Proposition \ref{prop:perfect_hedgeability}, where we have $C^{(t)}_s = e^{-\bfi t}$ for all $s$ and $t$.

\subsection{Hedging Method 2}

In this method $\bD_t =D_{\tau}^{t/\tau}$ for $t\in\mE$, and since it is nonlinear in $D_{\tau}$, by Proposition \ref{prop:perfect_hedgeability} there can be no perfect hedge. Turning then to first order hedges, we have $\bz_t := z_{\tau}$, so $\delta\bz_t[z|\Dz] = \Dz_{\tau}$. Here the hedge equation is
\begin{align*}
\int_\mE t\Dz_t\,dA^*_t = \Dz_{\tau}\int_\mE t\,dL^*_t.
\end{align*}
This holds if $dA^*_t = \one\{t=\tau\}\frac{1}{\tau}\int_\mE t\,dL^*_t$. The hedge has a lump sum at the LLP whose dollar duration $\tau\,dA^*_{\tau}$ equals that of the dollar duration of all liabilities with times to maturities $t\in\mE$: $\int_\mE t\,dL^*_t$.

If this method were to be mandated for all insurance companies, it would put a lot of buying preasure on the zero coupon bond with maturity $\tau$ since that is needed to hedge all longer liabilities. This could in turn lower $z_{\tau}$ which would increase the value of the liabilities, and this could necessitate even further hedging by companies who had not been fully hedged previously, driving the yield even lower.

It is also worth noting that the market value of the hedge is larger than the present value of the liabilities. In order words, this means that a premium equal to the present value of liabilities is not enough to buy the required hedge, and an insurance company would have to resort to leverage.

Another issue with the method is that the first order hedge is lacking convexity compared to the liabilities. This is seen by inspecting \eqref{diff2P} for this method. The second order variation of the hedge is
\begin{align*}
\int_\mT t^2(\Dz_t)^2\,dA^*_t &= \tau(\Dz_{\tau})^2\int_\mE t\,dL^*.
\end{align*}
Since $\delta^2\bz[z|\Dz]=0$, the second order variation of the liabilities is
\begin{align*}
\int_\mE \left(t^2(\delta\bz_t[z|\Dz])^2 - t\delta^2\bz_t[z|\Dz]\right)\,dL^*_t &= (\Dz_{\tau})^2\int_\mE t^2\,dL^*_t.
\end{align*}
The difference of the second order variations is
$$
\tau(\Dz_{\tau})^2\int_\mE t\,dL^* - (\Dz_{\tau})^2\int_\mE t^2\,dL^*_t = -(\Dz_{\tau})^2\int_\mE t(t-\tau)\,dL^*_t < 0.
$$
The consequence of this is that the hedge must be increased regardless of whether the yield increases or decreses. This variable exposure could in practice be achieved by buying options --- both calls and puts --- on the zero coupon bond used for hedging, so an introduction of this method could also conceivably increase option prices (implied volatilities).

\subsection{Hedging Method 3}

For this method $\bD_t = e^{-(t-\tau)\bar f_\infty}D_\tau$, which is linear in $D_{\tau}$, and it is therefore possible to hedge this method perfectly, by taking $C^{(t)}_s = e^{-(t-\tau)\bar f_\infty}\one\{s\geq\tau\}$. It is also instructive to see the first order properties of the hedge.
$$
\bz_t = \frac{\tau}{t}z_{\tau}+(1-\frac{\tau}{t})\bar f_{\infty}
$$
for $t\in\mE$ and therefore
$$
\delta\bz_t[z|\Dz] = \frac{\tau}{t}\Dz_{\tau}.
$$
We arrive at the hedge equation
\begin{align*}
\int_\mE t\Dz_t\,dA^*_t = \tau\Dz_{\tau}\int_\mE dL^*_t.
\end{align*}
The solution is $dA^*_t = \one\{t=\tau\}\int_\mE dL^*_t=\one\{t=\tau\}L^*_T$. Similarly as for Method 2 the hedge consists of a lump sum at the LLP. The difference is that the market value of the lump sum should equal the present value of all liabilities at times $s\geq \tau$, i.e.\ $\int_\mE dL^*_t$, whereas the hedge for Method 2 requires matching of dollar durations. 

Since $\int_\mE dL^*_t < \frac{1}{\tau}\int_\mE t\,dL^*_t$, the hedge for Method 3 requires less than Method 2 to be invested at the LLP $\tau$. The problem with all hedgers wanting to invest in the zero coupon bond with maturity $\tau$ would thus be diminished.

Also, since the market value of the hedge equals the present value of the liability, no leverage is needed.

\subsection{(Impossibility of) hedging Method 4}

Here
$$
\bz_t = \frac{\tau}{t}z_{\tau}+(1-\frac{\tau}{t})f_{\tau}
$$
for $t\in\mE$, so
$$
\delta\bz_t[z|\Dz] = \frac{\tau}{t}\Dz_{\tau}+(1-\frac{\tau}{t})\Delta f_{\tau}
$$
and we get the hedge equation
\begin{align*}
\int_\mE t\Delta z_t\,dA^*_t &= \tau\Delta z_{\tau}\int_\mE dL^*_t +\Delta f_{\tau}\int_\mE (t-\tau)\,dL^*_t.
\end{align*}
We recognize the first term on the right hand side from the hedge equation of Method 3, and know how to hedge that. The last term with $\Delta f_{\tau}$ is troublesome. A proper hedge of that would require an exposure to the forward rate at time $\tau$ and no other yields at no other times.

A \emph{forward rate agreement} (FRA) is a derivative that provides exposure to the forward rate over a certain interval, say $\tau-\epsilon$ to $\tau$. It can be replicated by the asset flow
\begin{align*}
dF_t &= \begin{cases} \frac{1}{\epsilon}, & t=\tau-\epsilon\\ -\frac{1}{\epsilon}e^{\int_{\tau-\epsilon}^{\tau}f_s\,ds}, & t = \tau, \end{cases} 
\end{align*}
i.e.\ one agrees today to borrow $\frac{1}{\epsilon}$ units at time $\tau-\epsilon$ and repay $\frac{1}{\epsilon}e^{\int_{\tau-\epsilon}^{\tau}f_s\,ds}$ units at time $\tau$. Note that the market value of these two flows cancel each other:
\begin{align*}
dF^*_{\tau-\epsilon} &=D_{\tau-\epsilon}\,dF_{\tau-\epsilon} = \frac{1}{\epsilon}D_{\tau-\epsilon} = \frac{1}{\epsilon}e^{-(\tau-\epsilon) z_{\tau-\epsilon}} = \frac{1}{\epsilon}e^{-\tau z_{\tau}}e^{\int_{\tau-\epsilon}^{\tau}f_s\,ds} \\
&=-e^{-\tau z_{\tau}}\,dF_{\tau}=-D_{\tau}\,dF_{\tau}=-dF^*_{\tau},
\end{align*}
so the market value of this contract is zero at inception: $F^*_T=\int_\mT dF^*_t=0$, and its interest rate sensitivity is
\begin{align*}
\delta F^*_T[z|\Dz] &= \int_\mT t\Dz_t\,dF^*_t = \frac{1}{\epsilon}D_{\tau}\big(-(\tau-\epsilon)\Dz_{\tau-\epsilon} + \tau\Dz_{\tau}\big) \\
&=D_{\tau}\frac{1}{\epsilon}\int_{\tau-\epsilon}^{\tau}\Delta f_s\,ds.
\end{align*}
Note that the value increases with increasing forward rate in contrast to bond values decreasing with increasing yield.

In order to isolate $\Delta f_{\tau}$ we would have to let $\epsilon\to 0$. However, that would mean that the amount borrowed, $\frac{1}{\epsilon}$, would diverge to infinity.

This problem is of course artificial in the sense that it appears due to us insisting on working with a continuum of maturities. In reality one would be exposed to the forward rate between the last two maturities used in constructing the yield curve. It could nonetheless be problematic, though not impossible, to hedge this exposure since the necessary forward rate agreement would entail one being short zero coupon bonds at the second to last maturity and long zero coupon bonds at the last maturity, and this might be hard to achieve in sufficient size, especially if the whole insurance industry wants to short the same bond.

\subsection{Hedging Method 5 (SFSA)}

This method is defined in terms of forward rates and as we show in Appendix \ref{z_method_5},
\begin{align*}
\bar z_t &= \begin{cases} \frac{\kappa-t}{\kappa-\tau}z_t + \frac{1}{t}\frac{1}{\kappa-\tau}\int_{\tau}^t sz_s\,ds + \frac{t-\tau}{\kappa-\tau}\big(1-\frac{\tau}{t}\big)\frac{\bar f_{\infty}}{2}, & \tau < t \leq \kappa, \\
                     \frac{1}{t}\frac{1}{\kappa-\tau}\int_{\tau}^{\kappa} sz_s\,ds + \big(1-\frac{\tau+\kappa}{2t}\big)\bar f_{\infty}, & t > \kappa.        \end{cases} 
\end{align*}
The corresponding discount factor $\bD_t$ is not affine in the market discount factors, so by Proposition \ref{prop:perfect_hedgeability} we can have no perfect hedge. However, first order hedgeability is possible. The G\^ateaux variation is
$$
\delta\bz_t[z|\Dz]  = \begin{cases} \frac{\kappa-t}{\kappa-\tau}\Dz_t + \frac{1}{t}\frac{1}{\kappa-\tau}\int_{\tau}^t s\Dz_s\,ds, & \tau < t \leq \kappa, \\
                     \frac{1}{t}\frac{1}{\kappa-\tau}\int_{\tau}^{\kappa} s\Dz_s\,ds, & t > \kappa,        \end{cases} 
$$
and the hedge equation is
\begin{align*}
\int_\mT \Dz_t\,dA^*_t &= \int_{\tau}^{\kappa}\frac{\kappa-t}{\kappa-\tau}t\Delta z_t\,dL^*_t + \int_{\tau}^{\kappa}\bigg(\frac{1}{\kappa-\tau}\int_{\tau}^t s\Delta z_s\,ds\bigg)\,dL^*_t \\
&\quad + \int_{\kappa}^{T}\bigg(\frac{1}{\kappa-\tau}\int_{\tau}^{\kappa} s\Delta z_s\,ds\bigg)\,dL^*_t \\
&= \int_{\tau}^{\kappa}\frac{\kappa-t}{\kappa-\tau}t\Delta z_t\,dL^*_t + \int_{\tau}^{\kappa}\bigg(\frac{1}{\kappa-\tau}\int_{t}^{\kappa}dL^*_s\bigg) t\Delta z_t\,dt \\
&\quad + \frac{1}{\kappa-\tau}\int_{\kappa}^{T}\,dL^*_s\int_{\tau}^{\kappa} t\Delta z_t\,dt \\
&= \int_{\tau}^{\kappa}\frac{\kappa-t}{\kappa-\tau}t\Delta z_t\,dL^*_t + \int_{\tau}^{\kappa}\bigg(\frac{1}{\kappa-\tau}\int_{t}^{T}dL^*_s\bigg) t\Delta z_t\,dt.
\end{align*}

The hedge is thus $dA^*_t = \frac{\kappa-t}{\kappa-\tau}\,dL^*_t + (\frac{1}{\kappa-\tau}\int_{t}^{T}dL^*_s)\,dt$ for $\tau<t\leq \kappa$.

The hedge consists of two terms. One, $\frac{\kappa-t}{\kappa-\tau}\,dL^*_t$, matches the corresponding liability cash flow at time $t$ to a linearly decreasing extent. The other, $(\frac{1}{\kappa-\tau}\int_{t}^{T}dL^*_s)\,dt = \frac{1}{\kappa-\tau}(L^*_T-L^*_t)\,dt$, consists of a flow that is proportional to the present value of all liabilities larger than or equal to $t$.

The extrapolation described by Method 5 linearly transitions from the market forward rate to the predetermined long term forward rate, in contrast to the instant transition of Method 3. Comparing the hedges for the two methods, we see that the hedge of Method 5 also linearly transitions from exactly hedging present value of each liability cash flow with the market value of an asset cash flow, to the ``edge case'' of Method 3 where the present value of all liabilities with higher maturity are hedged by a single lower duration bond.

Since the hedging is not concentrated to a single bond, but to the whole interval of bonds with maturities between $\tau$ and $\kappa$, this method would probably be even less disruptive to the bond market than Method 3.

The total market value of the hedge equals the present value of the liabilities:
\begin{align*}
\int_{\tau}^{\kappa}dA^*_t &= \int_{\tau}^{\kappa}\frac{\kappa-t}{\kappa-\tau}\,dL^*_t + \frac{1}{\kappa-\tau}\int_{\tau}^{\kappa}(L^*_T-L^*_t)\,dt \\
&= \left[\frac{\kappa-t}{\kappa-\tau}L^*_t\right]_{t=\tau}^{t=\kappa}+ \frac{1}{\kappa-\tau}\int_{\tau}^{\kappa}L^*_t\,dt+L^*_T-\frac{1}{\kappa-\tau}\int_{\tau}^{\kappa}dL^*_t\\ 
&= L^*_T,
\end{align*}
and thus, similarly to Method 3, leverage is not needed.

Since this method is not perfect as Method 3, second order properties of the hedge are worthwhile to investigate. As shown in Appendix \ref{cnvx_method_5}, in contrast to Method 2, the hedge actually has an excess of convexity compared to the liabilities. For parallel shifts of the yield curve $z$ between times $\tau$ and $\kappa$, regardless of whether these are up or down, the hedge would need to be decreased. If an insurance company would like to capitalise on this, it could --- again in contrast to Method 2 --- sell options, and this could push down option prices.

\subsection{(Impossibility of) hedging Method 6 (Smith-Wilson)}

The discrete version of the Smith-Wilson method is clearly perfectly hedgeable since $\bm{\zeta}$ is linear in $\bm{D}=(D_{t_1},\dots,D_{t_N})$, and thus any $\bar D_t$ is too. However, the continuous case is not even first order hedgeable, and the cause points to practical problems even with the ``perfect'' hedge.

Recall \eqref{sw_z}:
$$
\bz_t = \frac{\tau}{t}z_{\tau}+\left(1-\frac{\tau}{t}\right)\bar f_{\infty}-\frac{1}{t}\log\left(1+(\bar f_{\infty}-f_{\tau})\frac{1-e^{-\alpha(t-\tau)}}{\alpha}\right),
$$
for $t\in\mE$ which leads to
$$
\delta\bz_t[z|\Dz] = \frac{\tau}{t}\Dz_{\tau}+c(t)\Delta f_{\tau}
$$
where
$$
c(t) = c(t;\tau,\alpha,f_{\tau},\bar f_{\infty}) := \frac{1-e^{-\alpha(t-\tau)}}{\alpha t}\bigg/\bigg(1+(\bar f_{\infty}-f_{\tau})\frac{1-e^{-\alpha(t-\tau)}}{\alpha}\bigg).
$$
The appearance of $\Delta f_{\tau}$ in the expression for the G\^ateaux variation implies the same problems for this method as for Method 4.

The similarity goes further: The parameter $\alpha$ is typically small, and when $\alpha\to 0$, $c(t)\to (1-\frac{\tau}{t})/(1+(\bar f_{\infty}-f_{\tau})(t-\tau))$, and since $\bar f_{\infty}$ and $f_{\tau}$ are also on the order of a few percentage points, $c(t)\approx (1-\frac{\tau}{t})$, which is the coefficient of $\Delta f_{\tau}$ in the G\^ateaux variation of Method 4.

An optimal hedge would thus need similar exposure as Method 4 to the forward rate between the last two maturities, as shown for particular discrete yield curves by Rebel \cite{cardano}.

In the variant of the Smith-Wilson method where $\alpha$ is chosen to ensure $|\bar f_{\kappa}-\bfi|\leq \epsilon$, $\alpha$ is itself a function of $z$ with parameters $\tau$, $\kappa$ and $\epsilon$. An implicit expression for $\alpha[z]$ can be derived from \eqref{sw_f}, but we will not pursue this here.

The G\^ateaux variation is now $\delta\bz_t[z|\Dz]+\frac{d\bz_t}{d\alpha}\delta\alpha[z|\Dz]$, where $\delta\bz_t[z|\Dz]$ is the variation in the Smith-Wilson method with fixed $\alpha$, $\frac{d\bz_t}{d\alpha}$ the derivative of the right hand side of \eqref{sw_z} with respect to $\alpha$, and $\delta\alpha[z|\Dz]$ is the first order variation of $\alpha[z]$. 

\subsection{Sensitivity with respect to the UFR}

If the extrapolation method changes, the value of the liabilities can change as well. It may not be possible to hedge against such model changes but knowledge of the sensitivity with respect to parameters could nevertheless be useful since the supervisory agency might want to introduce a new method or change the parameters of the current method.

The sensitivity with respect to the UFR can easily be calculated for the particular methods described above. We introduce $S:=-\frac{dL^*_T}{d\bfi}\big/L^*_T$, which can be seen as a duration with respect to the UFR.

We continue to assume $L_{\tau}=0$.

\subsubsection{Method 2}

Here $\frac{d}{d\bfi}\bD_t = \frac{d}{d\bfi}e^{-\bfi t} = -t e^{-\bfi t} = -t\bD_t$ and thus
$$
S = \frac{\int_\mE t\,dL^*_t}{L^*_T} = \dur[\bz,L],
$$
i.e.\ the duration of $L$.

\subsubsection{Method 3}

Method 3 has $\frac{d}{d\bfi}\bD_t = -(t-\tau)\bD_t$ and thus
$$
S = \frac{\int_\mE (t-\tau)\,dL^*_t}{L^*_T} = \dur[\bz,L]-\tau,
$$
which is the duration of $L$ \emph{minus $\tau$}, so its sensitivity is strictly less than that of Method 2. We call this quantity the \emph{excess duration above $\tau$};
$$
\excdur[y,C,\tau]:=\frac{\int_{\tau}^T (t-\tau)\,dC^*_t[y]}{C^*_T[y]}.
$$

\subsubsection{Method 5}

For Method 5,
\begin{align*}
\frac{d}{d\bfi}(t\bz_t) &= \begin{cases}  \frac{(t-\tau)^2}{2(\kappa-\tau)}, & \tau < t \leq \kappa, \\
                     t-\frac{\tau+\kappa}{2}, & t > \kappa.        \end{cases} 
\end{align*}
so
\begin{align*}
\frac{d}{d\bfi}\bD_t &= \frac{d}{d\bfi}e^{-t\bz_t} = -\frac{d}{d\bfi}(t\bz_t)\cdot\bD_t= \begin{cases}  -\frac{(t-\tau)^2}{2(\kappa-\tau)}\bD_t, & \tau < t \leq \kappa, \\
                     -\big(t-\frac{\tau+\kappa}{2}\big)\bD_t, & t > \kappa.        \end{cases} 
\end{align*}
and
\begin{align*}
-\frac{d}{d\bfi}L^*_T &= \int_{\tau}^{\kappa}\frac{(t-\tau)^2}{2(\kappa-\tau)}\,dL^*_t + \int_{\kappa}^{T}\left(t-\frac{\tau+\kappa}{2}\right)\,dL^*_t \\
&=\frac{1}{2}\bigg[\underbrace{\int_{\tau}^{\kappa}\frac{(t-\tau)^2}{\kappa-\tau}\,dL^*_t}_{=:I}+\int_{\kappa}^T (t-\tau)\,dL^*_t + \int_{\kappa}^T (t-\kappa)\,dL^*_t\bigg].\\
I &\leq \int_{\tau}^{\kappa}(t-\tau)\,dL^*_t \Rightarrow \\
S &\leq \frac{\excdur[\bz,L,\tau]+\excdur[\bz,L,\kappa]}{2}.\\
I &\geq 0\Rightarrow \\
S  &\geq \frac{(\kappa-\tau)L^*_{\kappa}}{2} + \excdur[\bz,L,\kappa].
\end{align*}
The sensitivity is thus less than the average of the excess durations above $\tau$ and $\kappa$, and we also have a lower bound that might be useful. Since $\excdur[\bz,L,\kappa]\leq \excdur[\bz,L,\tau]$, the sensitivity is always lower than that of Method 3.

\subsubsection{Method 6}

For the Smith-Wilson method,
\begin{align*}
-\frac{d}{d\bfi}\bD_t &= (t-\tau)\bD_t  - e^{-\bfi(t-\tau)}\frac{1-e^{-\alpha(t-\tau)}}{\alpha}D_\tau \\
&= (t-\tau)\bD_t - \frac{1}{\alpha}\big(\bD^{(1)}_t - \bD^{(\alpha)}_t \big) \\
&\to (t-\tau)\bD_t - (t-\tau)\bD^{(0)},
\end{align*}
as $\alpha\to 0$, where $\bD^{(\alpha)}_t$ is the discount factor of Method 3 with $\bfi+\alpha$ as UFR, instead of $\bfi$. Let $\bz^{(\alpha)}$ and be the corresponding discount curve, with $\bz$ still denoting the Smith-Wilson curve. We also have $\lim_{\alpha\to\infty}\frac{d}{d\bfi}\bD_t = -(t-\tau)\bD_t$. Taken altogether, 
\begin{align*}
S &= \excdur[\bz,L,\tau]-\frac{L^{*}_T[\bz^{(0)}] - L^{*}_T[\bz^{(\alpha)}]}{\alpha L^*[\bz]} \\
S &\geq \excdur[\bz,L,\tau] - \excdur[\bz^{(0)},L,\tau] \\
S &\leq \excdur[\bz,L,\tau],
\end{align*}
with $S$ approaching the bounds as $\alpha$ tends to 0 or $\infty$. The Smith-Wilson method also has a sensitivity that is less than that of Method 3. How it compares to Method 5 depends on the value of $\alpha$ and the liability cash flow.

\section{Conclusions and future research}\label{sec:conclusions}

We have presented a framework that can be used to derive an optimal first order hedge. In essence we have generalised the common practice of matching key rate durations at a finite number of times to maturity to a continuum of times to maturity. The advantage with this generalisation is that it allows for easy comparison between different extrapolation methods.

Among the extrapolation methods to which we have applied the framework, we find that some methods, including the Smith-Wilson method, need a hedge with exposure to the forward rate at the last liquid point, which is troublesome to execute in the market since it necessitates shorting of the penultimate liquid market point. Other methods, such as constant extrapolation of zero coupon yields (Method 2) or a discontinuous transition to a prespecified ultimate forward rate (Method 3) only requires long exposure to the last liquid point. However, Method 3 requires less exposure than Method 2, and would therefore be less dispruptive to the market as a whole. The method mandated by the SFSA is even less dispruptive to the market since it entails a gradual transition from market rates to predetermined rates. Its downside compared to Method 3 is that it only allows for a first order hedge whereas Method 3 can be hedged perfectly.

Future research could investigate other extrapolation methods such as the Nelson-Siegel-Svensson method. The framework could also be used to investigate the hedgeability of bootstrap methods, i.e.\ methods for interpolating liability discount factor from market observations.

This paper has only focused on instantaneous changed in yield curves. As time passes, maturities of the hedging instruments decreases from $\tau$ (for all methods except Method 5), and if one wants to hold a bond with maturity $\tau$, the existing hedge has to be sold and a new with longer maturity bought. Another area for future studies is how this would affect companies and fixed income markets.

\section*{Acknowledgement}

I thank Martin Bender, Mathias Lindholm and Jan Svedberg for comments on a draft of this paper.

\appendix
\section{Method 5}

\subsection{Discount yield}\label{z_method_5}

Here we derive $\bz$ for Method 5 (Section \ref{method_5}).

For $\tau < t \leq \kappa$,
\begin{align*}
t\bar z_t &= \int_0^{\tau}\bar f_s\,ds + \int_{\tau}^{t} \bar f_s\,ds = \tau z_{\tau} + \int_{\tau}^t \bar f_s\,ds \\
&= \tau z_{\tau} + \frac{1}{\kappa-\tau}\int_{\tau}^t(\kappa-s)f_s\,ds + \frac{1}{\kappa-\tau}\int_{\tau}^t(s-\tau)\bar f_{\infty}\,ds\\
&= \frac{\tau\kappa-\tau^2}{\kappa-\tau}z_{\tau} + \frac{\kappa}{\kappa-\tau}(t z_t-\tau z_{\tau}) - \frac{1}{\kappa-\tau}\int_{\tau}^ts f_s\,ds + \frac{(t-\tau)^2}{2(\kappa-\tau)}\bar f_{\infty}\\
&= \frac{\kappa t z_t-\tau^2 z_{\tau}}{\kappa-\tau}- \frac{1}{\kappa-\tau}\int_{\tau}^ts f_s\,ds + \frac{(t-\tau)^2}{2(\kappa-\tau)}\bar f_{\infty}\\
&= \frac{\kappa t z_t-\tau^2 z_{\tau}}{\kappa-\tau}- \frac{1}{\kappa-\tau}\bigg[s^2z_s\bigg]_{s=\tau}^{s=t}+\frac{1}{\kappa-\tau}\int_{\tau}^tsz_s\,ds+ \frac{(t-\tau)^2}{2(\kappa-\tau)}\bar f_{\infty} \\
&= \frac{\kappa-t}{\kappa-\tau}t z_t+\frac{1}{\kappa-\tau}\int_{\tau}^ts z_s\,ds+ \frac{(t-\tau)^2}{2(\kappa-\tau)}\bar f_{\infty}.
\end{align*}

For $t\geq \kappa$,
\begin{align*}
tz_t &= \int_0^{\kappa}\bar f_s\,ds + \int_{\kappa}^{t} \bar f_s\,ds = \kappa\bar z_{\kappa} + \int_{\kappa}^t\bar f_{\infty}\,ds = \kappa z_{\kappa} + (t-\kappa)\bar f_{\infty} \\
&= \frac{1}{\kappa-\tau}\int_{\tau}^{\kappa}s z_s\,ds+\frac{\kappa-\tau}{2}\bar f_{\infty}+(t-\kappa)\bar f_{\infty} \\
&= \frac{1}{\kappa-\tau}\int_{\tau}^{\kappa}s\bz(s)\,ds+\bigg(t-\frac{\tau+\kappa}{2}\bigg)\bar f_{\infty}.
\end{align*}

\subsection{Second order properties}\label{cnvx_method_5}

Let $L$ be a liability cash flow and $A$ the corresponding first order hedge (at $z$) of Method 5 (Section \ref{method_5}). Let $L_{\tau}=0$ in order to focus on the extrapolated part of the yield curve. Assume $\Dz_t = 1$ for $\tau\leq t\leq \kappa$. In this subsection we will show that
$$
\int_\mT t^2(\Delta z_t)^2dA^*_t > \int_\mE \left(t^2(\delta\bz_t[z|\Dz])^2 - t\delta^2\bz_t[z|\Dz]\right)\,dL^*_t.
$$
It suffices to show this for all lump sum liabilities $L_t^*=\one\{t\geq \sigma\}$ with corresponding hedge $A^*_t$.

To ease notation, let
\begin{align*}
a(\sigma) &:= \int_\mT t^2(\Dz_t)^2dA^*_t, \\
l(\sigma) &:= \int_\mE \left(t^2(\delta\bz_t[z|\Dz])^2 - t\delta^2\bz_t[z|\Dz]\right)\,dL^*_t.
\end{align*}
We have for $\tau<\sigma\leq\kappa$
$$
dA^*_t = \begin{cases} \frac{dt}{\kappa-\tau}, & \tau<t\leq \sigma \\
\frac{\kappa-\sigma}{\kappa-\tau}, & t=\sigma\end{cases}
$$
and for $\sigma>\kappa$,
$$
dA^*_t = \frac{dt}{\kappa-\tau},\quad\text{for $\tau<t\leq\kappa$}.
$$

We have
\begin{align*}
\delta\bz_t[z|\Dz]&=\begin{cases} \frac{\kappa-t}{\kappa-\tau}+\frac{1}{t}\frac{1}{\kappa-\tau}\int_{\tau}^t s\,ds,&\tau<t\leq\kappa,\\ \frac{1}{t}\frac{1}{\kappa-\tau}\int_{\tau}^{\kappa} s\,ds,&t>\kappa,\end{cases} \\
&=\begin{cases} \frac{\kappa^2-\tau^2-(\kappa-t)^2}{2t(\kappa-\tau)},&\tau<t\leq\kappa,\\ \frac{\kappa^2-\tau^2}{2t(\kappa-\tau)},&t>\kappa,\end{cases}\\
&=\begin{cases} \frac{\kappa^2-\tau^2-(\kappa-t)^2}{2t(\kappa-\tau)},&\tau<t\leq\kappa,\\ \frac{\kappa+\tau}{2t},&t>\kappa,\end{cases}
\end{align*}
and $\delta^2\bz[z|\Dz]=0$.

Consider first the case $\sigma>\kappa$.
\begin{align*}
a(\sigma) &= \int_{\tau}^{\kappa}\frac{t^2}{\kappa-\tau}\,dt =\frac{\kappa^3-\tau^3}{3(\kappa-\tau)}=\frac{\kappa^2+\kappa\tau+\tau^3}{3} \\
l(\sigma) &= \left(\frac{\kappa+\tau}{2}\right)^2 = \frac{\kappa^2+2\kappa\tau+\tau^2}{4} \\
\Rightarrow\quad a(\sigma) - l(\sigma) &= \frac{(\kappa-\tau)^2}{12}>0.
\end{align*}

Now consider the case $\tau<\sigma\leq\kappa$, and introduce $\lambda:=\frac{\sigma-\tau}{\kappa-\tau}$.
\begin{align*}
a(\sigma) &= \int_{\tau}^{\sigma}\frac{t^2}{\kappa-\tau}\,dt+\sigma^2\frac{\kappa-\sigma}{\kappa-\tau}\\
l(\sigma) &= \frac{(\kappa^2-\tau^2-(\kappa-\sigma)^2)^2}{4(\kappa-\tau)^2}.
\end{align*}
We note that $\lim_{\sigma\to\tau}a(\sigma)=\lim_{\sigma\to\tau}l(\sigma)=\tau^2$.
\begin{align*}
a'(\sigma) &= 2\sigma\frac{\kappa-\sigma}{\kappa-\tau}=2\sigma(1-\lambda).\\
l'(\sigma) &= \frac{(\kappa^2-\tau^2-(\kappa-\sigma)^2)(\kappa-\sigma)}{(\kappa-\tau)^2}=(\tau+\lambda\kappa+(1-\lambda)\sigma)(1-\lambda) \\
a'(\sigma)-l'(\sigma) &= (2\sigma-\tau-\lambda\kappa-(1-\lambda)\sigma)(1-\lambda) \\
&= (\sigma-\tau-\lambda(\kappa-\sigma))(1-\lambda) \\
&= (\kappa-\tau)(\lambda-\lambda(1-\lambda))(1-\lambda) = (\kappa-\tau)\lambda^2(1-\lambda)>0.
\end{align*}
Since $\lim_{\sigma\to\tau}a(\sigma)-l(\sigma) = 0$ and $a'(\sigma)-l'(\sigma)>0$, $a(\sigma)-l(\sigma)>0$ also for $\tau<\sigma\leq\kappa$, and we are done.

\subsection{Sensitivity with respect to UFR}\label{ufr_sens_method_5}

\begin{align*}
\frac{d}{d\bfi}(t\bz_t) &= \begin{cases}  \frac{(t-\tau)^2}{2(\kappa-\tau)}, & \tau < t \leq \kappa, \\
                     t-\frac{\tau+\kappa}{2}, & t > \kappa.        \end{cases} 
\end{align*}
so
\begin{align*}
\frac{d}{d\bfi}\bD_t &= \frac{d}{d\bfi}e^{-t\bz_t} = -\frac{d}{d\bfi}(t\bz_t)\cdot\bD_t= \begin{cases}  -\frac{(t-\tau)^2}{2(\kappa-\tau)}\bD_t, & \tau < t \leq \kappa, \\
                     -\big(t-\frac{\tau+\kappa}{2}\big)\bD_t, & t > \kappa.        \end{cases} 
\end{align*}
and
\begin{align*}
-\frac{d}{d\bfi}\bP[z;L] &= \int_{\tau}^{\kappa}\frac{(t-\tau)^2}{2(\kappa-\tau)}\,dL^*_t + \int_{\kappa}^{T}\left(t-\frac{\tau+\kappa}{2}\right)\,dL^*_t \\
&=\frac{1}{2}\bigg[\int_{\tau}^{\kappa}\frac{(t-\tau)^2}{\kappa-\tau}\,dL^*_t+\int_{\kappa}^T (t-\tau)\,dL^*_t + \int_{\kappa}^T (t-\kappa)\,dL^*_t\bigg].\\
&\leq \frac{1}{2}\bigg[\int_{\tau}^{\kappa}(t-\tau)\,dL^*_t+\int_{\tau}^T (t-\tau)\,dL^*_t + \int_{\kappa}^T (t-\kappa)\,dL^*_t\bigg]. \\
&= \frac{1}{2}\bigg[\int_{\tau}^T (t-\tau)\,dL^*_t + \int_{\kappa}^T (t-\kappa)\,dL^*_t\bigg].
\end{align*}
The sensitivity is thus less than the average of the excess durations above $\tau$ and $\kappa$.

\section{Method 6}

\subsection{Continuous version}\label{z_method_6}


Andersson and Lindholm \cite{sw_rep} have derived the following representation of the Smith-Wilson discount factors. Essentially, it shows that a Smith-Wilson discount factor for time to maturity $t$ is the expected value of a Gaussian process at time $t$ conditioned on its values at times $t_i$, $i=1,\dots,N$, being the observed zero coupon prices for time to maturities $t_i$, $i=1,\dots,N$.

\begin{thm}
Let $X$ be a Gaussian Ornstein-Uhlenbeck process with stochastic differential $dX_t=-\alpha X_t\,dt + \alpha^{3/2}\,dB_t$ and initial value $X_0\sim N(0,\alpha^2)$ independent of $B$ ; $X^*_t:=\int_0^t X_s\,ds$; and $Y_t:=e^{-\bfi t}(1+X^*_t)$. Then
$$
\bar D_t = \Ex[Y_t|Y_{t_i} = D_{t_i}; i=1,\dots,N].
$$
\end{thm}
Note that the Ornstein-Uhlenbeck process in the theorem is not stationary; for stationarity $X_0$ should be $N(0,\frac{\alpha^2}{2})$. We have presented the theorem above in a more streamlined form than Andersson and Lindholm, and we therefore also provide a proof.
\begin{proof}
We have
$$
X_t = X_0e^{-\alpha t}+\alpha^{3/2}\int_{\tau}^te^{-\alpha(t-s)}\,dB_s = e^{-\alpha t}\bigg(X_0+\alpha^{3/2}\int_{\tau}^te^{\alpha s}\,dB_s\bigg)
$$
Thus, $\Ex[X_t]=0$ for all $t$, which imply $\Ex[X^*_t]=0$, and hence $\Ex[Y_t]=e^{-\bfi t}$. For $s\leq t$,
\begin{align*}
\Cov(X_s,X_t) &= e^{-\alpha(s+t)}\Cov(e^{\alpha s}X_s,e^{\alpha t}X_t) \\
&= e^{-\alpha(s+t)}\Cov\bigg(X_0+\alpha^{3/2}\int_{\tau}^s e^{\alpha u}\,dB_u,X_0+\alpha^{3/2}\int_{\tau}^t e^{\alpha u}\,dB_u\bigg) \\
&= e^{-\alpha(s+t)}\bigg(\Var(X_0) + \alpha^3\int_{\tau}^s e^{2\alpha u}\,du\bigg) \\
&= e^{-\alpha(s+t)}\bigg(\alpha^2+\frac{\alpha^2}{2}(e^{2\alpha s}-1)\bigg) \\
&= \alpha^2 e^{-\alpha t} \cosh(\alpha s).
\end{align*}
This in turn gives us, with $s\leq t$,
\begin{align*}
\Cov(X^*_s,X^*_t) &= \iint\limits_{0\leq u\leq s \atop 0\leq v\leq t}\Cov(X_u,X_v)\,dudv \\
&= \bigg(\iint\limits_{0\leq u\leq v \leq s}+\iint\limits_{0\leq v \leq u\leq s}+\iint\limits_{0\leq u\leq s \leq v\leq t}\bigg)\Cov(X_u,X_v)\,dudv \\
&= \bigg(2\iint\limits_{0\leq u\leq v \leq s}+\iint\limits_{0\leq u\leq s \leq v\leq t}\bigg)\Cov(X_u,X_v)\,dudv \\
&= 2\int_0^s\alpha e^{-\alpha v}\bigg(\int_0^v\alpha\cosh(\alpha u)\,du\bigg)\,dv \\
&\quad + \int_s^t\alpha e^{-\alpha v}dv\int_0^s\alpha\cosh(\alpha u)\,du \\
&= 2\int_0^s\alpha e^{-\alpha v}\sinh(\alpha v)\,dv + (e^{-\alpha s}-e^{-\alpha t})\sinh(\alpha s)\\
&= \alpha s - e^{-\alpha t}\sinh(\alpha s).
\end{align*}
Finally we can deduce
\begin{align*}
\Cov(Y_s,Y_t) &= e^{-\bfi(s+t)}\Cov(X^*_s,X^*_t) = W(s,t).
\end{align*}
Let $\bm{Y}:=(Y_{\tau},\dots,Y_{t_N})'$ and $\bm{D}:=(D_{\tau},\dots,D_{t_N})'$. Since $Y$ is a Gaussian process
\begin{align*}
\Ex[Y_t|\bm{Y}=\bm{D}] &= \Ex[Y_t] + \Cov(Y_t,\bm{Y})\Cov(\bm{Y},\bm{Y})^{-1}(\bm{D}-\Ex[\bm{Y}]) \\
&= e^{-\bfi t} + \Cov(Y_t,\bm{Y})\bm{\zeta} \\
&= e^{-\bfi t} + \sum_{i=1}^N\Cov(Y_t,Y_{t_i})\zeta_i \\
&= e^{-\bfi t} + \sum_{i=1}^N W(t,t_i)\zeta_i,
\end{align*}
as desired, where we have identified $\bm{\zeta}:=(\zeta_1,\dots,\zeta_N)':=\Cov(\bm{Y},\bm{Y})^{-1}(\bm{D}-\Ex[\bm{Y}])$.
\end{proof}


We will use this representation to obtain a continuous version of the method when the curve is fitted to all zero coupon bonds with time to maturities between 0 and $\tau$. We define the continuous version of the Smith-Wilson method thus:
\begin{equation}\label{sw_cont1}
\bar D_t := \Ex[Y_t|Y_s = D_s; 0\leq s \leq \tau],
\end{equation}
where $Y_t$ is the Gaussian process of the theorem. Before we try to simplify this expression we must learn more about the stochastic processes of the theorem.

Note that the process $X^*$ is not Markov whereas the augmented process $(X^*,X)$ is. For $t\geq s$ we have
\begin{align*}
X^*_t &= X^*_s+\int_s^tX_u\,du \\
&= X^*_s+\int_s^t\bigg(X_s e^{-\alpha(u-s)}+\alpha^{3/2}\int_s^ue^{-\alpha(u-v)}\,dB_v\bigg)\,du \\
&= X^*_s+X_s\frac{1-e^{-\alpha(t-s)}}{\alpha}+\alpha^{3/2}\int_s^t\int_s^ue^{-\alpha(u-v)}\,dB_v\,du,
\end{align*}
and in particular
$$
\Ex[X^*_t|X^*_s=x^*,X_s=x] = x^*+x\frac{1-e^{-\alpha(t-s)}}{\alpha}.
$$
Also, the sigma-algebra generated by $(X^*_s,X_s)$ is a subset of the sigma-algebra generated by $\{X^*_u;0\leq u\leq s\}$, i.e.\ if we know the whole trajectory of $X^*$ up to time $s$ we know both its value \emph{and its derivative} at time $s$. Thus, for $t\geq s$,
$$
\Ex[X^*_t|X^*_u = x^*_u; 0\leq u\leq s] = \Ex[X^*_t|X^*_s=x^*_s,X_s=x_s] = x^*_s+x_s\frac{1-e^{-\alpha(t-s)}}{\alpha},
$$
where $x_s := \frac{d}{ds}x^*_s$.

Let us now return to the definition \eqref{sw_cont1}. For $0\leq t\leq \tau$, we clearly have $\bar D_t= D_t$, i.e.\ $\bar z_t = z_t$. Let $x^*_t$ be defined by $D_t=e^{-tz_t} =: e^{-\bfi t}(1+x^*_t)$ and let $x_t := \frac{d}{dt}x^*_t$. We have
$$
x^*_t = e^{\bfi t}D_t-1 = e^{\bfi t-\int_0^t f_s\,ds}-1,
$$
so
$$
x_t = (\bfi - f_t)e^{\bfi t}D_t.
$$

For $t\geq \tau$,
\begin{align*}
e^{-t\bar z_t}= D_t &= \Ex[Y_t|Y_s = D_s; 0\leq s\leq \tau] \\
&= \Ex[e^{-\bfi t}(1+X^*_t)|X^*_s=x^*_s; 0\leq s\leq \tau] \\
&= e^{-\bfi t}(1+\Ex[X^*_t|X^*_{\tau}=x^*_{\tau}, X_{\tau}=x_{\tau}]) \\
&= e^{-\bfi t}\bigg(1+x^*_{\tau}+x_{\tau}\frac{1-e^{-\alpha(t-\tau)}}{\alpha}\bigg)\\
&= e^{-\bfi(t-\tau)}D_{\tau}\bigg(1+(\bfi - f_{\tau})\frac{1-e^{-\alpha(t-\tau)}}{\alpha}\bigg).
\end{align*}

\subsection{Example of arbitrage for Smith-Wilson}\label{arb}

If one fits a Smith-Wilson curve to one zero coupon bond with yield 0, yields will start negative. Indeed, 
\begin{align*}
D_{t_1} &= \bar D_{t_1} \\
\iff\quad 1 &= e^{-\bar f_{\infty}t_1}+W(t_1,t_1)\zeta_1 \\
&= e^{-\bar f_{\infty}t_1} + e^{-2\bar f_{\infty}t_1}(\alpha t_1- e^{-\alpha t_1}\sinh(\alpha t_1))\zeta_1 \\
\Rightarrow\quad \zeta_1 &= e^{\bar f_{\infty}t_1}\frac{e^{\bar f_{\infty}t_1}-1}{\alpha t_1- e^{-\alpha t_1}\sinh(\alpha t_1)} \\
\bar D_t &= e^{-\bar f_{\infty}t}\left( 1 + \frac{\alpha t - e^{-\alpha t_1}\sinh(\alpha t)}{\alpha t_1- e^{-\alpha t_1}\sinh(\alpha t_1)}(e^{\bar f_{\infty}t_1}-1) \right),\quad 0\leq t\leq t_1 \\
\Rightarrow \frac{d}{dt} \bar D_t |_{t=0} &= -\bar f_{\infty} + \underbrace{\frac{\alpha - \alpha e^{-\alpha t_1}}{\alpha t_1- e^{-\alpha t_1}\sinh(\alpha t_1)}}_{>\frac{1}{t_1}}\underbrace{(e^{\bar f_{\infty}t_1}-1)}_{>\bar f_{\infty}t_1} > 0.
\end{align*}
Hence $\bar D_t$ is \emph{increasing} at $t=0$ and since $\bar D_0=1$ we thus have discount factors greater than one for small $t$, which corresponds to negative yields.

In reality --- as opposed to the idealised setting of Mathematical Finance --- negative yields only imply an arbitrage if its feasible to hold cash at zero cost, which is not the case for larger amounts of cash, since they would have to be put in a guarded vault. That the Smith-Wilson method can produce negative yields might therefore not be such a big problem in practice.


\begin{thebibliography}{9}

\bibitem{sw_rep} \textsc{Andersson, H.\ and Lindholm, M.} (2013). On the relation between the Smith-Wilson method and integrated Ornstein-Uhlenbeck processes. Research report in Mathematical Statistics 2013:1, \emph{Dept.\ of Mathematics, Stockholm University}.

\bibitem{eiopa} \textsc{European Insurance and Occupational Pensions Authority} (2014). Stress Test 2014: Technical Specification for the Preparatory Phase (Part II). \url{https://eiopa.europa.eu/fileadmin/tx_dam/files/publications/technical_specifications/B_-_Technical_Specification_for_the_Preparatory_Phase__Part_II_.pdf}

\bibitem{nashed} \textsc{Nashed, M.\ Z.} (1966). Some remarks on variations and differentials. \emph{Amer.\ Math.\ Monthly}. \textbf{73.}4, part 2., 63--76.

\bibitem{cardano} \textsc{Rebel, L.} (2012). The Ultimate Forward Rate Methodology to Valuate Pensions Liabilities: A Review and an Alternative Methodology. Research report, \emph{Cardano}. \url{http://www.netspar.nl/files/Evenementen/2012-11-09\%20PD/rebel.pdf}

\bibitem{sagan} \textsc{Sagan, H.} (1992). \emph{Introduction to the Calculus of Variations}. Dover Publications, New York. Reprint of 1969 original.

\bibitem{sw} \textsc{Smith, A.\ and Wilson, T.} (2000). Fitting Yield Curves with Long Term Constraints. Research report, \emph{Bacon \& Woodrow}.

\bibitem{fi} \textsc{Swedish Financial Supervisory Authority} (2013). Föreskrifter och allmänna råd om försäkringsföretags val av räntesats för att beräkna försäkringstekniska avsättningar. FFFS 2013:23. \url{http://fi.se/Regler/FIs-forfattningar/Samtliga-forfattningar/201323/}


\end{thebibliography}
\end{document}